\documentclass{article}
\usepackage[margin=1.3in]{geometry}

\usepackage{authblk}
\usepackage[round]{natbib}

\usepackage{booktabs}
\usepackage{amsfonts}
\usepackage{graphicx}
\usepackage{amsmath,amssymb,amsthm,bm,enumerate}
\usepackage{latexsym,color,verbatim,minipage-marginpar,caption,multirow}
\usepackage{mycommands}
\usepackage[lined,ruled]{algorithm2e}

\newcommand{\hW}{\widehat{W}}
\newcommand{\hA}{\widehat{A}}
\newcommand{\hT}{\widehat{T}}
\newcommand{\hcP}{\widehat{\mathcal{P}}}
\newcommand{\hEE}{\widehat{\mathbb{E}}}
\newcommand{\cM}{\mathcal{M}}
\newcommand{\Nr}{\mathcal{N}}
\newcommand{\corr}{\mathrm{corr}}
\renewcommand{\Cov}{\mathrm{Cov}}
\newcommand{\WFcomment}[1]{{\color{red}{(WF:  #1) }}}
\newcommand{\DTcomment}[1]{{\color{blue}{(DT:  #1) }}}
\renewcommand{\WFcomment}[1]{}
\renewcommand{\DTcomment}[1]{}
\newcommand{\op}{\mathrm{op}}

\graphicspath{ {./figures/} }

\title{Family Learning: Non-Parametric Statistical Inference with Parametric Efficiency}
\author{William Fithian\thanks{Much of the work was done while a contractor at Facebook.}}
\affil{\texttt{wfithian@berkeley.edu}\\
Department of Statistics, University of California, Berkeley}
\author{Daniel Ting\thanks{Work was done while employed at Facebook.}}
\affil{\texttt{dting@tableau.com}\\
Tableau Software}

\begin{document}

\newtheorem{theorem}{Theorem}
\newtheorem{corollary}[theorem]{Corollary}
\newtheorem{lemma}[theorem]{Lemma}
\newtheorem{observation}[theorem]{Observation}
\newtheorem{proposition}[theorem]{Proposition}
\newtheorem{definition}[theorem]{Definition}
\newtheorem{claim}[theorem]{Claim}
\newtheorem{fact}[theorem]{Fact}
\newtheorem{assumption}{Assumption}

\newcommand{\bX}{\mathbf{X}}
\newcommand{\bS}{\mathbf{S}}
\newcommand{\bD}{\mathbf{D}}
\newcommand{\bT}{\mathbf{T}}
\newcommand{\bU}{\mathbf{U}}
\newcommand{\bV}{\mathbf{V}}

\newcommand{\cX}{\mathcal{X}}
\newcommand{\cH}{\mathcal{H}}
\newcommand{\convL}{\ensuremath{\stackrel{L_2}{\to}}}

\maketitle

\begin{abstract}
Hypothesis testing and other statistical inference procedures are most efficient when a reliable low-dimensional parametric family can be specified. We propose a method that learns such a family when one exists but its form is not known {\em a priori}, by examining samples from related populations and fitting a low-dimensional exponential family that approximates all the samples as well as possible. We propose a computationally efficient spectral method that allows us to carry out hypothesis tests that are valid whether or not the fit is good, and recover asymptotically optimal power if it is. Our method is computationally efficient and can produce substantial power gains in simulation and real-world A/B testing data.
\end{abstract}

\section{Introduction}\label{sec:intro}

In most statistical problems, parametric modeling assumptions bring many benefits: they simplify the problem, guide us in deriving efficient inference procedures, and help us to summarize and understand the data parsimoniously. These benefits, however, must be weighed against the risk that if we choose a badly misspecified model, the resulting inferences may be highly misleading.

This article proposes a method to {\em learn} a good parametric model from the data when we face a repeated task analyzing many statistical experiments with similar underlying characteristics. Our motivating insight is that we can treat similar experiments as arising from a common parametric model, each one corresponding to a different parameter value. If a smooth parametric model exists, our method successfully estimates and exploits it to derive an optimal hypothesis testing procedure for the next experiment, asymptotically achieving the same efficiency as if the model were known {\em a priori}. In addition, these tests retain their advertised signficance level under nonparametric assumptions.

\subsection{A/B Testing}\label{sec:AB}

The original motivation for this work is {\em A/B testing}: randomized experiments conducted at internet companies to evaluate user behavior under various treatment conditions, typically minor modifications of the website's behavior. Most large internet companies carry out many A/B tests each day as a way to test new features or check that routine updates have not broken the site.  Companies also want to understand how people are engaging with their products so that they can make informed and effective decisions about how to improve upon them. They do so by examining metrics about user engagement such as the time spent on the site or number of photos liked. The repeated nature of this task presents an opportunity for the company to learn from past experiments in order to improve their inferences for the next one.

Despite a common misconception that statistical power is always extremely high when analyzing internet-scale data sets, statistical efficiency is in fact at a high premium in A/B testing for several reasons. First, many of the user metrics that firms use for evaluation have characteristically heavy right tails, making nonparametric estimators like the sample mean highly inefficient. Second, the changes observed can be small, but many small incremental changes can have a cumulatively large effect. For example, detecting an additional 200 changes that each raise user engagement by 0.2\% would increase engagement by 50\% overall. Third, it is important that the methods are powerful enough to withstand slicing the data into subgroups for more detailed analyses, or correcting for multiple comparisons across many experiments. Finally, companies prefer to involve as few users as possible in each experiment. 

\vspace{-0.2cm}
\subsection{Notation and Problem Setting}\label{sec:notation}

We consider an A/B testing problem with $m$ independent treatment groups, obtained by combining many experiments with several treatment groups each, all drawn from a common user population and measuring the same outcome of interest. The $i$th treatment group is $\smash{X_{i,1},\ldots, X_{i,n_i} \simiid P_i}$, with $n_+= \sum_i n_i$ denoting the total number of observations.

We introduce a further working assumption that the distributions $P_i$ arise from a common smooth parametric family of densities $\cP = \{p_{\theta}:\; \theta \in \Omega \sub \R^k\}$ with respect to a common dominating measure $\nu$ on the outcome space $\cX$; typically $\cX\sub \R$ and $\nu$ is the usual Lebesgue measure. We emphasize that the functional form of $p_\theta$ is unknown. Let $\E_\theta[\cdot], \Var_\theta[\cdot], \corr_\theta(\cdot,\cdot)$ etc. denote moments computed with respect to $p_\theta$, and let $\theta_i$ denote the parameter for group $i$ ($dP_i = p_{\theta_i}d\nu$).

We consider a local asymptotic regime in which the sample sizes grow while the effect sizes shrink, a good match for Facebook data. That is, we assume $\Omega$ contains an open neighborhood of 0 and $\theta_i\to 0$. Locally, it will be sufficient to estimate a local approximation $p_\theta \approx p_0(1 + \theta'T)$, where $T(x)$ is the efficient score function at $\theta=0$. We assume the  regularity condition that $\cP$ is {\em differentiable in quadratic mean} (DQM) at $\theta=0$ and that the Fisher information $\E_0 TT'$ is non-singular. 
Under standard regularity conditions the score is more familiarly the gradient of the log-likelihood $T(x) = \nabla_\theta \left.\log p_{\theta}(x)\right|_{\theta=0}$. DQM is a weaker condition not requiring differentiability of $\log p_\theta$ at $0$. While $\E_0 T(X) = 0$ by necessity, $T(x)$ is only identifiable up to invertible linear transformations. To resolve the ambiguity, we assume without loss of generality that $\Var_0 T(X) = I_k$ (even then, $T(X)$ is only identifiable up to orthogonal transformations).

Given estimates $\hT$ and $\hat{p}_0$, we approximate $\cP$ using the exponential family model $\hcP = \{e^{\theta'\hT(x) - \psi(\theta)}\hat{p}_0(x):\; \int_\cX e^{\theta'\hT}\hat{p}_0 < \infty\}$. As we will see, applying the estimated model $\smash{\hcP}$ to future experiments gives nearly the same statistical efficiency as knowing the actual parametric model. Because we are  estimating an exponential family approximation, we will refer to $\smash{\hT(x)}$ as the {\em sufficient statistic} and $\smash{\hat{p}_0(x)}$ as the {\em base measure}, reflecting the roles they play in $\smash{\hcP}$.

As a concrete example, consider the Laplace location family, with $p_{\theta}(x) = \smash{\frac{1}{2}} e^{-|x-\theta|}$. This is not an exponential family, but $\smash{p_{\theta}(x) \approx p_0(x) e^{\theta \sgn(x)}}$ for small values of $\theta$. Therefore $\cP$ is locally well-approximated by an exponential family with sufficient statistic $T(x) = \sgn(x)$. Figure~\ref{fig:laplace} shows the results of a simulation (revisited in Section~\ref{sec:laplace}) that illustrates our method. By estimating $T(x)$ and $p_0(x)$, we obtain a one-parameter exponential family that we can subsequently use to analyze experiments as though we had prior knowledge of the true parametric form.

\begin{figure}
  \centering
  \begin{tabular}{ccc}
    \includegraphics[width=.28\textwidth, height=.28\textwidth]{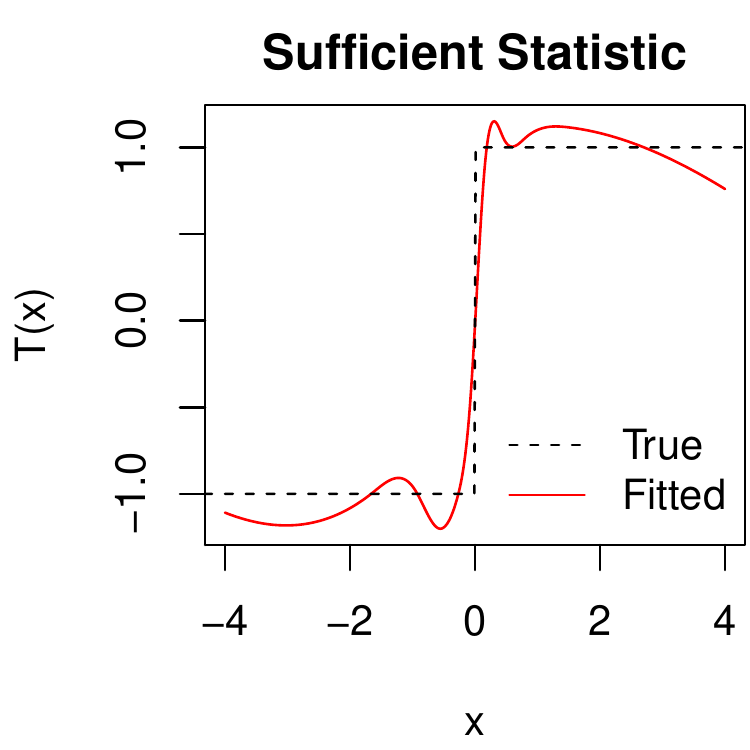} 
    & 
      \includegraphics[width=.28\textwidth, height=.28\textwidth]{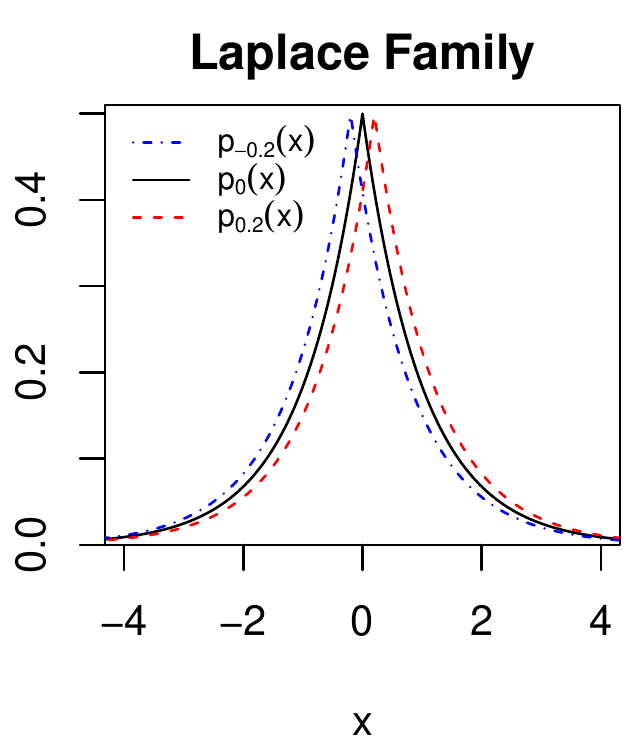}		
    & 
      \includegraphics[width=.28\textwidth, height=.28\textwidth]{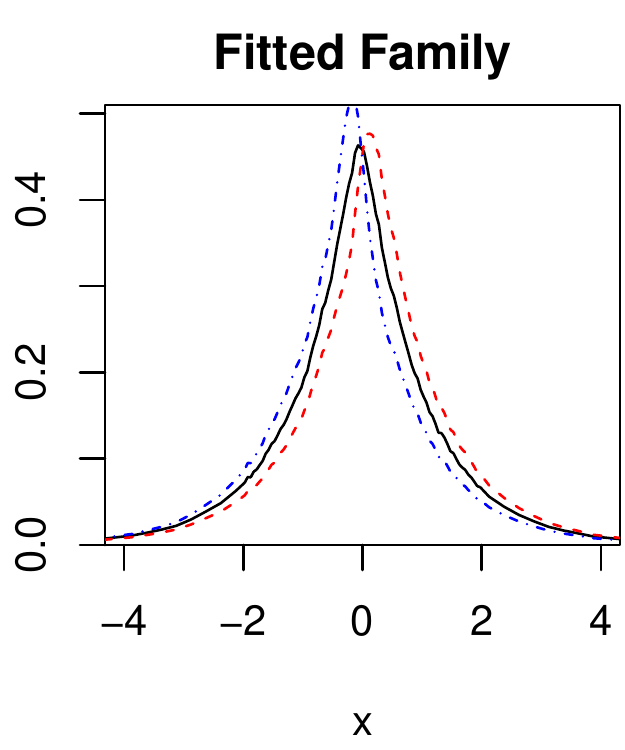}		
  \end{tabular}
  \vspace{-.3cm}
  \caption{Left: Fitted sufficient statistic for Laplace simulation. Middle: Laplace density for $\theta = -0.2, 0,$ and $0.2$. Right: Fitted densities for the same parameter values.
  }
  \label{fig:laplace}
\end{figure}

Section~\ref{sec:est} details our method for estimation of $\cP$ and proves consistency. Section~\ref{sec:inference} discusses how to perform inference once we have our estimate $\smash{\hcP}$, and Sections~\ref{sec:simulations}--\ref{sec:facebook} illustrate and evaluate our method using simulated and real data from A/B tests performed at Facebook. Section~\ref{sec:discussion} concludes.

\subsection{DQM Families and Efficient Hypothesis Testing}\label{sec:dqm}
While there are many possible uses for the estimated family $\hcP$, we focus on hypothesis testing in a new experiment, i.e. testing $H_0:\; P_{m+1} = P_{m+2}$. In particular, for any candidate transformation $U(x)$ with variance $\sigma_U^2<\infty$ we could use the test statistic
\begin{align} \label{eq:test statistic}
\Delta_U = \sqrt{\frac{n_{m+1}n_{m+2}}{(n_{m+1}+n_{m+2})\hat\sigma_U^2}}
\left[\frac{1}{n_{m+1}}\sum_{j=1}^{n_{m+1}}U(X_{m+1,j}) - \frac{1}{n_{m+2}}\sum_{j=1}^{n_{m+2}}U(X_{m+2,j})\right],
\end{align}
where $\hat\sigma_U^2$ is the pooled sample variance. We can find the critical value either by permutation (giving an exact test in finite samples) or by using the more computationally convenient normal approximation $\Delta_U \stackrel{H_0}{\to} \Nr(0,1)$, which gives the right significance level in large samples. 
These procedures are fully non-parametric in the sense that they yield an asymptotically correct critical value without making any parametric assumptions on the underlying distribution.

We now briefly review the asymptotic theory of hypothesis testing in the case $k=1$; for a full treatment see~\citet{van2000asymptotic}. The differentiable in quadratic mean (DQM) condition is that the square root of the density can be approximated as $\sqrt{p_\theta(x)} = \sqrt{p_0(x)}\left(1+\theta'T(x)/2\right) + r_\theta(x)$, where the remainder term satisfies $\| r_\theta \|_{L_2(\nu)}^2 = \int_\cX r_\theta(x)^2\,d\nu(x)= o(\|\theta\|^2)$. 
The DQM assumption is a weak and easily satisfied condition and includes ill-behaved distributions such as power law distributions. 
If the experiment's samples are from a DQM family $\cP$ with score $T(X)$ at 0, then the (one-sided) test based on $\Delta_T$ is asymptotically equivalent to the uniformly most powerful likelihood ratio test. We analyze power in the local asymptotic regime, where the parameters for the two samples becomes closer as the  sample size grows, $\theta_{m+1} = t_{m+1}/\sqrt{n_{m+1}}$ and $\theta_{m+2} = t_{m+2}/\sqrt{n_{m+2}}$; this scaling ensures that the test is neither powerless nor perfectly powerful in the limit as $n_{m+1},n_{m+2}\to\infty$. Under this scaling, Le Cam's Third Lemma \WFcomment{original reference?} implies that for almost any candidate $U$, we have 
\begin{equation}\label{eq:varDelta}
\Delta_U \to \Nr\left(\frac{t_{m+1}-t_{m+2}}{2}\corr_0(U,T), 1\right).
\end{equation}
We make this relationship more explicit in section \ref{sec:inference}.
Note that if $\corr_0(U,T) = 1/2$, we would need to quadruple the size of our sample to achieve the same power using the test statistic $\Delta_U$ that we could have achieved by using $\Delta_T$ instead. In that sense, the test based on $\Delta_U$ uses the data only $25\%$ as efficiently as the test based on $\Delta_T$; we say its {\em Pitman relative efficiency} is 0.25 relative to the better test \WFcomment{original reference?}, meaning that testing with $\Delta_U$ is asymptotically equivalent to first discarding $75\%$ of the data and then testing with $\Delta_T$. Indeed, Le Cam's Third Lemma implies that the Pitman relative efficiency of a test based on {\em any} asymptotically normal test statistic can be measured by its squared correlation to $\Delta_T$ under $p_0$, and
the score $T$ defines the asymptotically efficient test statistic.

\WFcomment{Something like ``There is a comparable theory for estimators and confidence intervals, see \citet{van2000asymptotic} for a full treatment.''}

\vspace{-.1cm}
\subsection{Related work}
\citet{hastie1987principal} propose a method to decompose contingency tables into low-dimensional multinomial families. Their goal closely resembles ours in spirit in the case of discrete data with a saturated basis $S$. To our knowledge, our work is the first to propose using a meta-analysis of low powered experiments in order to obtain procedures that are asymptotically efficient, and in particular, to do so non-parametrically. 
There is a rich body of work on non-parametric two-sample tests in a single experiment.
Some tests, such as the t, Mann-Whitney, or Wilcoxon tests, are not guaranteed to have power even in the limit of infinite data, while others such as the Kolmogorov-Smirnov and Anderson-Darling tests, 
the Kernel Two-Sample Test of \cite{gretton2012test}, and the B-test of \cite{zaremba2013btest} do. For finite samples 
 choosing the appropriate test can be difficult as it is not clear which test is most powerful for a given problem. 

Unlike most non-parametric hypothesis tests, our method also allows for inference about functionals $\phi(F)$ of a distribution. These include the mean, median, or other interpretable quantities of interest. Semi-parametric methods such as those by \cite{vanDerLaan2006tmle} and \cite{klaassen1987consistent} can be used when there is 
target functional or parameter of interest without resorting to strong parametric assumptions. Surprisingly, they can give asymptotically efficient estimators. 
However, the efficiency depends on eventually having a high signal-to-noise ratio, an unrealistic assumption
when the goal is to improve inference for experiments with low power. We instead combine many experiments with low signal to improve efficiency. 

\cite{fithian2015heavytailed} and \cite{taddy2016heavytailed} propose methods 
to improve mean estimates in A/B testing by using a large background dataset to better estimate the tails for heavy tailed distributions.

\vspace{-.2cm}
\section{Estimating the Family $\cP$}\label{sec:est}

Our method for estimating $\cP$ proceeds in two steps. First, we employ a spectral method to estimate $T(x)$, which requires no knowledge or estimate of $p_0(x)$. Second, we use the fitted $\smash{\hT(x)}$ to estimate $p_0(x)$ using a Poisson generalized linear model (GLM). For simplicity we focus our analysis on the one-parameter case; $k>1$ is similar but with more notational overhead.
\vspace{-.2cm}
\subsection{Estimating the Sufficient Statistic $T(x)$}\label{sec:estT}

Let $(S_i(x))_{i=1}^\infty$ denote a dense basis of $L^2(p_0)$, specified in advance and with each $S_i$ bounded. For a fixed dimension $d$, we approximate $T$ as a linear combination of the first $d$ basis vectors $S(x) = (S_1(x), \ldots, S_d(x))$. Truncating the basis to $d$ elements incurs an approximation error shrinking to 0 as $d\to\infty$.

Let $W = \Var_0 S(X)\in\R^{d\times d}$. Define the group mean $\bar S_i  = \frac{1}{n_i} \sum_{j=1}^{n_i} S(X_{i,j})$, the grand mean $\bar S = \frac{1}{n_+} \sum_{i=1}^m n_i \bar S_i$, and the empirical between- and within-treatment-group covariance matrices:
\begin{align}
\hA &= \frac{1}{m-1}\sum_{i=1}^m n_i(\bar S_i - \bar S)(\bar S_i - \bar S)', \quad \text{ and}\\
\hW &= \frac{1}{n_+-m} \sum_{i=1}^m \sum_{j=1}^{n_i} (S(X_{i,j}) - \bar S_i) (S(X_{i,j}) - \bar S_i)'.
\end{align}
We estimate $\hT(x)=\hat\beta'S(x)$, where $\hat\beta\in\R^d$ maximizes the Rayleigh quotient $\smash{\beta'\hA\beta / \beta'\hW\beta}$.

We now state our main theorem: under appropriate assumptions, $\smash{\hT}$ converges to $T$ in $L^2(p_0)$, up to translation and scaling. Consider a sequence of problems satisfying the following:

\begin{assumption}\label{assu:degrees}
$m, n_+-m \to \infty$ (the total degrees of freedom for $\hW$ and $\hA$ are growing)
\end{assumption}

\begin{assumption}\label{assu:signal}
$\frac{1}{m}\sum_{i=1}^m n_i(\theta_i-\bar\theta)^2 = \Omega(1)$ (the average signal strength does not shrink to 0)
\end{assumption}

\begin{assumption}\label{assu:smallEffects}
$\max_i |\theta_i| \to 0$ (effect sizes are small)
\end{assumption}

Assumptions~\ref{assu:degrees}--\ref{assu:smallEffects} are fairly mild conditions. To grasp them more easily, consider a balanced random effects model where $n_i = n$, and $\theta_i = t_i/\sqrt{n}$ where $t_i\simiid \Nr(\alpha, \sigma^2)$ with $\alpha,\sigma^2$ fixed. Then, Assumptions~\ref{assu:degrees}--\ref{assu:smallEffects} hold almost surely provided that $m,n\to\infty$ with $\log m = o(n)$.

\begin{theorem}\label{thm:consistencyT}
  Under Assumptions~\ref{assu:degrees}--\ref{assu:smallEffects}, if we take $d \to\infty$ slowly enough, then $\hT(x)$ is consistent for $T(x)$ in $L^2(p_0)$ up to sign; that is, 
  $\corr_0^2(\hT(X), T(X) | \hT) \toProb 1$.
\end{theorem}

This proof and all others are deferred to the technical supplement, but the following derivation motivates $\hat\beta$ heuristically: Let $\mu_0 = \E_0 S(X)$, $c = \Cov_0(S(X), T(X))$, $D = \E_0[TSS']$, and write
\vspace{-.1cm}
\begin{align}
\mu_i &\;=\; \E_{\theta_i} S(X) 
\;\approx\; \int_{\cX} S(x) p_0(x)(1 + \theta_i'T(x)) \, d\nu(x)
\;=\; \mu_0 + \theta_i c,\\
W_i &\;=\; \Cov_{\theta_i} S(X)
\;\approx\; W + \theta_i (D + \mu_0 c' + c \mu_0') + \theta_i^2 cc' 
= W + O(\theta_i)
\end{align}

Writing $T_\beta(x) = \beta'S(x)$, we would ideally choose $\beta$ to maximize correlation with the score:
\vspace{-.1cm}
\[
\rho_\beta^2  = \frac{\Cov_0(T_\beta(X), T(X))^2}{\Var_0(T_\beta(X))} = \frac{\beta'cc'\beta}{\beta'W\beta}.
\]
Since $c$ and $W$ are unknown, we estimate $\hW \approx W$ and maximize a proxy for $\rho_\beta^2$. If $\bar\theta = n_+^{-1}\sum_i n_i \theta_i$ then the expectation of the between-groups covariance matrix $\smash{\hA}$ is 
\begin{equation}\label{eq:Astar}
  A = \E \hA \;\approx\; \frac{1}{m-1}\sum_{i=1}^m n_i (\theta_i-\bar\theta)^2 cc' + \left(1 -\frac{n_i}{n_+} \right)W_i
       \;=\; \frac{\tau^2}{m-1} cc' + W + O(\max_j |\theta_j|),
\end{equation}
where $\tau^2 = \sum_{i=1}^m n_i(\theta_i-\bar\theta)^2$ is the overall signal strength, leading to the objective function:
\begin{align} \label{eqn:objectiveFunction}
\hat f(\beta) &\;=\; \frac{\beta'\hA\beta}{\beta'\hW\beta} 
\;\approx\; \frac{\beta'A\beta}{\beta'W\beta} = \frac{\tau^2\rho_\beta^2}{m-1} + 1 = f(\beta)
\end{align}
$\hat f(\beta)$ is maximized by $\hat\beta=\hW^{-1/2}\tilde\beta$, where $\tilde\beta$ is the leading eigenvector of $\hW^{-1/2}\hA\;\hW^{-1/2}$.

Note that Theorem~\ref{thm:consistencyT} does not tell us how fast to take $d\to\infty$, much less how to choose $d$ in finite samples. A more refined analysis would reveal that the choice depends greatly on quantities like the signal strength $\tau^2$, the correlation under $p_0$ of the basis elements, and the smoothness of $\cP$; since these are all unknown, in practice we recommend choosing $d$ via cross-validation: since Type I error on held-out data is always $\alpha$, we can choose $d$ to maximize the number of rejections. Cross-validation can also help choose between competing bases. In our motivating problems we have found natural spline bases with $d=10$ to $15$ sufficient to achieve very high correlations with the true $T(x)$. 

The assumption that $S$ is bounded can be weakened to subgaussianity and mild moment assumptions, but these assumptions are unverifiable without knowing $\cP$ (unless $S$ is bounded). Still, if practical domain knowledge motivates using a given unbounded basis in a given setting, there is nothing in the theory to suggest our method will fail. Once again, cross-validation is the best guide.

Our method is highly computationally efficient. Most of the work is computing the summary statistics $\bar S_1,\ldots, \bar S_m \in \R^d$ and $\hW \in \R^{d\times d}$, which can be done in a single parallelizable scan over the entire data set, in time $\mathcal{O}(n_+d^2)$. The remaining computation to obtain $\hT$ requires $\mathcal{O}(m d^2 + d^3)$ time.

{
	\subsection{Estimating the Base Measure $p_0(x)$}

To produce the plots in Figure~\ref{fig:laplace}, we use a method following \citet{efron1996using}: we first bin the data, computing $N_{i,b} = \#\{j:\; X_{i,j} \in \text{bin } b\}$, and choosing for each bin a representative value $x_b = \E_{0}[X \mid X\in \text{bin } b]$. Using $\E N_{i,b} \approx n_i e^{\theta T(x)} p_0(x_b)$, we then estimate a Poisson GLM with 
$\log \E N_{i,b} = \alpha_i + \zeta_b + \theta_i \hT(x_b)$
 and set  $\hat p_0(x) \propto \sum_b e^{\hat\zeta_b} \delta_{x_b}(x)$.
}  

\section{Inference with an Estimated Family}\label{sec:inference}

After obtaining $\hcP$, we  want to use it for future experiments. The main result in this section is that, if $\corr_0(\hT(X), T(X) |\hT)^2 \to 1$, then hypothesis tests in future experiments are asymptotically efficient; in other words, we can do as well as if we knew the true family of distributions $\cP$ ahead of time.  Theorem \ref{thm:optimality} considers two-sample testing in a held-out experiment as discussed in Section \ref{sec:dqm}.

\begin{theorem}\label{thm:optimality}
	Let $\cP=\{p_{\theta}:\; \theta\in\Omega \sub \R\}$ be DQM at $\theta=0\in\Omega^o$, with score function $T(x)$ at $\theta=0$, and consider a sequence of estimators $\hT_n$ with $\corr_0(\hT_n(X), T(X) | \hT_n) \toProb 1$ as $n \to \infty$, independently of the two samples. If $n_{m+1},n_{m+2}\to \infty$ with $\theta_{m+1}=t_{m+1}/\sqrt{n_{m+1}}$ and $\theta_{m+2}=t_{m+2}/\sqrt{n_{m+2}}$, then the test rejecting for large $\Delta_{\hT_n}$ is asymptotically efficient for testing the one-sided hypothesis $H_0:\; \theta_{m+1} \leq \theta_{m+2}$ versus $H_1:\; \theta_{m+1} > \theta_{m+2}$.
\end{theorem}

Thus, our procedure produces an asymptotically efficient test as long as the appropriate sign can be correctly identified. This is typically simple as one is interested in the mean or some monotone increasing transformation of the metric. One  chooses the sign which makes $\corr(\hat{T}(X), X) > 0$.

The notion of asymptotic efficiency for multidimensional families is more complex as there is no uniformly most powerful test. However, one may still obtain asymptotically efficient tests when examining only one direction of the parameter space. In particular, for testing a difference of means, if one tests for a difference in $\nabla_0 \mu(\theta) (\theta_1 - \theta_0)$ where $\mu(\theta)$ maps the parameters to the mean, the resulting test is asymptotically efficient for testing a one-sided hypothesis for the difference in means.

\subsection{Other Inferences} \label{sec:otherInf}

The uses to which we can put our fitted family extend far beyond hypothesis testing. Generally speaking, we can use $\hcP$ in any of the myriad ways that we can use a parametric model that is specified in advance. The benefit of $\hcP$ over a prespecified model, however, if that we can proceed with the confidence that our parametric ``assumptions'' are actually validated by the data rather than having been chosen by fiat for convenience's sake.

For example, suppose that we are given a sample $X_{0,1},\ldots,X_{0,n_0}$ from some new population $p_{\theta_0}(x)$, and we wish to estimate some new parameter $\mu(\theta_0)$, say $\mu=\E_{\theta_0} X$. Given a confidence set $C_{\theta_0}$, we can derive a corresponding confidence set for $\mu$ via $C_\mu = \{\hEE_\theta X :\; \theta \in C_{\theta_0}\}$. In particular, if $\mu$ and $\theta_0$ are univariate and $C_{\theta_0}$ is an interval, and $\smash{\frac{d\mu}{d\theta}(0) \neq 0}$, then for small values of $\theta$ we can simply transform the endpoints of $C_{\theta_0}$ to obtain the endpoints of $C_\mu$. We apply this method in Section~\ref{sec:simulations} to obtain $95\%$ intervals for $\mu=\E_\theta X$, and see that the intervals enjoy their advertised coverage.

\section{Simulations}\label{sec:simulations}

\subsection{Laplace Location Family}\label{sec:laplace}

We revisit the Laplace location family discussed in Section~\ref{sec:notation}. We use a natural spline basis with 11 degrees of freedom. Note that the smooth spline basis is intentionally somewhat mismatched to the score $T(x)=\sgn(x)$ that we are trying to learn. This reflects the fact that we cannot rely on choosing a ``good'' basis in advance for real data. We generate $m=5000$ groups with $n_i=100$ data points each, with $\theta_i \sim \Nr(0,0.1^2)$. 

Despite the mismatched basis, our method learns the score up to correlation $\rho^2=93.8\%$. The $95\%$ confidence intervals for $\theta_i$ achieve their advertised coverage, with 248 of 5000 failing to cover. The average interval width is 0.4, which is $30\%$ shorter on average than one-sample $t$-intervals. Although $n_i$ is much too small to reliably distinguish $\theta=0$ from $\theta=0.1$ --- and would be even if we knew the true parametric form in advance --- by combining many groups together we can still learn much richer information, piecing together a very accurate local approximation to the family $\cP$.
 
\subsection{Log-Gamma Family}\label{sec:loggamma}
Next, we evaluate our method on the log-gamma family. That is, $\log(X_{i,j})$ is gamma-distributed with scale parameter $0.4$ and shape parameter $\theta_i$. Thus, population $i$ has exponential family density
\[
p_{\theta_i}(x) = e^{\theta_i \log(\log(x))} \cdot 
\frac{x^{-3.5} \log(x)^{-1}}{ 0.4^{\theta_i}\, \Gamma(\theta_i)},
\quad \text{ for } x>1.
\]
We have chosen the log-gamma family because of its heavy right tail, making it qualitatively similar to some real world metrics. We simulate $m=100$ groups with $n_i=5000$ observations each, with parameters $\theta_i$ drawn from $\Nr(3, 0.5^2)$, and use a natural spline basis with 11 degrees of freedom.

Figure~\ref{fig:loggamma} shows results from the simulation. In the left panel we plot the fitted sufficient statistic against the (centered and scaled) true score, $1.6\log(\log(x)) + 0.014$. Intuitively, $T(x)$ is forced to be highly concave because of the log-gamma distribution's heavy right tail --- for example, $T(x)=x$ is completely out of the question because $p_0(x)e^{\ep x}$ is not even normalizable for any $\ep>0$. Thus, the sample mean is a very poor summary of the data with relative efficiency only $7.7\%$; in other words, tests based on the sample mean effectively discard over $90\%$ of the data. By contrast, our method obtains an accurate estimate that is $98.6\%$ efficient.

The right two panels of Figure~\ref{fig:loggamma} show $95\%$ confidence intervals for the expectation $\mu_i = \E_{\theta_i}[X]  = (1-0.4)^{-\theta_i}$ of each treatment group computed with our method, alongside $t$-intervals for the same parameters. The intervals using the fitted family are much shorter despite retaining the correct advertised coverage; this is because they are based on a much more efficient summary of the data.

\begin{figure} 
  \centering
  \begin{tabular}{ccc}
    \includegraphics[width=.28\textwidth, height=.28\textwidth]{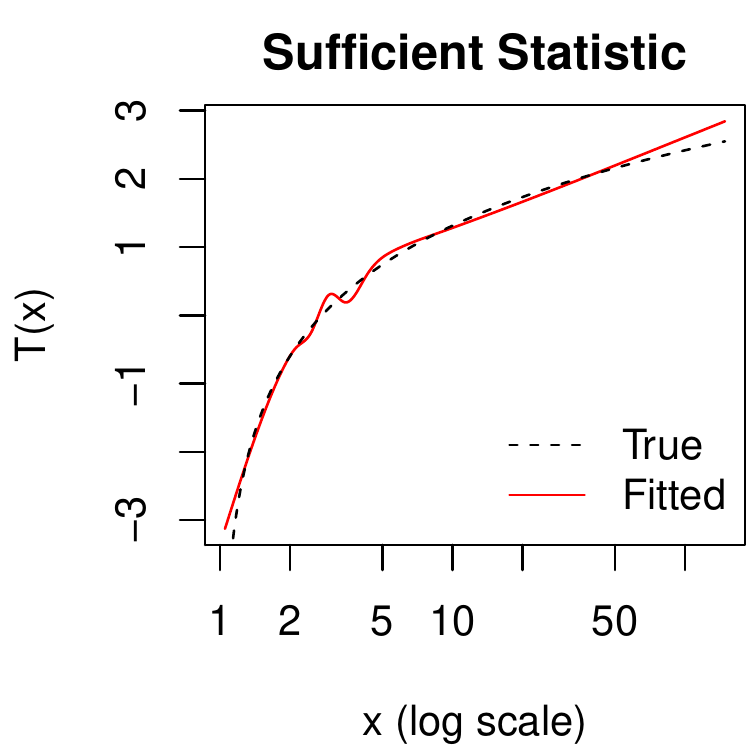} 
    & 
      \includegraphics[width=.28\textwidth, height=.28\textwidth]{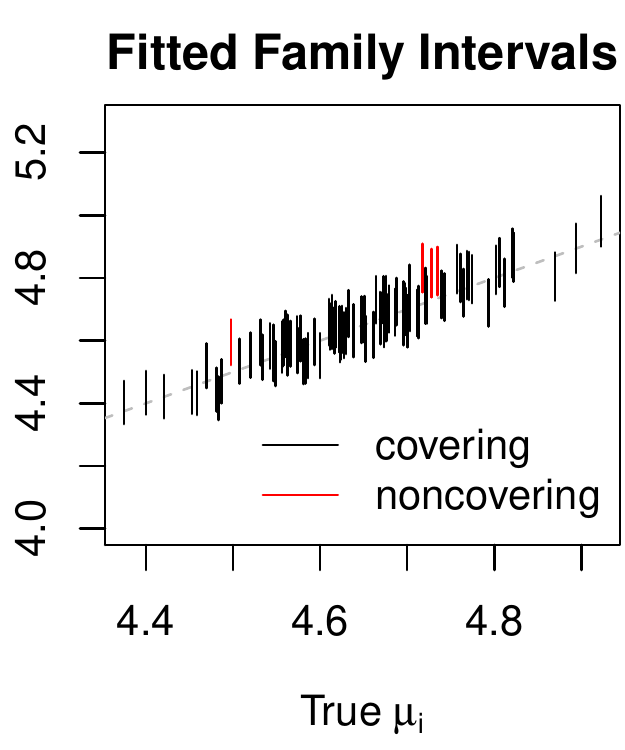}		
    & 
      \includegraphics[width=.28\textwidth, height=.28\textwidth]{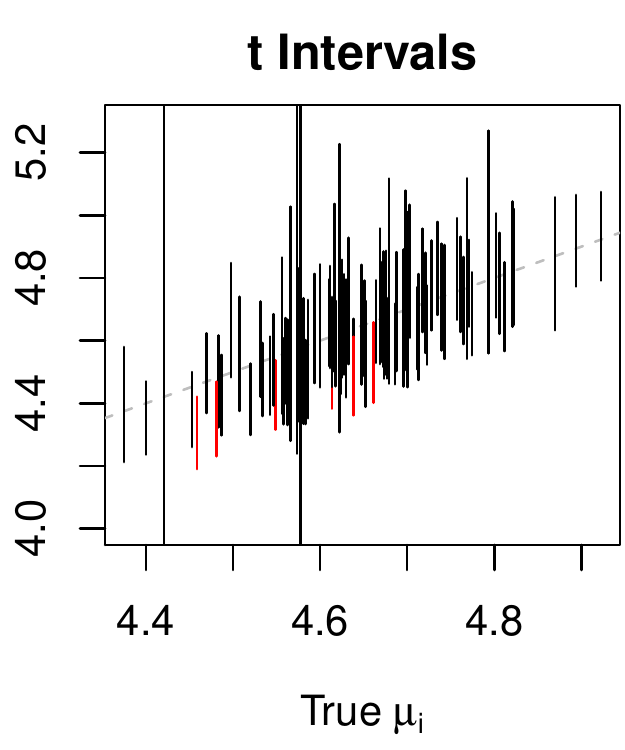}		
  \end{tabular}
\vspace{-0.1cm}
  \caption{Log-gamma results. Left panel: The fitted sufficient statistic is highly correlated with the true score $(\rho^2 = 98.6\%)$. Right two panels: $95\%$ confidence intervals for each $\mu_i$ using the fitted family, versus $t$-intervals for same.  The horizontal axis shows the true target $\mu_i$ for each interval. Both methods achieve the advertised coverage, but the $t$-intervals are $2.7$ times as long on average.}
  \label{fig:loggamma}  
\end{figure}

\section{Results on Facebook Data}\label{sec:facebook}

To give an example of our method in action, we demonstrate it on
Facebook A/B testing data, which provided the motivation for devising
our method. Some changes to Facebook result in small, incremental improvements to metrics of interest. These small differences are difficult to detect and necessitate either larger experiments or more effective measurements. The empirical results bear out that experimental effects behave locally like a low dimensional parameterized family. 
Furthermore, we are able to obtain a significantly more powerful test. 
Table \ref{tbl:RE} shows the estimated relative efficiencies of various statistics compared to using the estimated sufficient statistic.

The data consists of de-identified, aggregated 
distributions of rescaled outcomes over large, random subsets of users to ensure privacy. In the examples here, the data consists of 195 subpopulations with average sample size exceeding $2\times 10^6$. The x-axes are rescaled by some constant and the far tails of the distribution are truncated beyond the $99^{th}$ percentile in the plots .
\begin{figure}
	\centering
	\begin{tabular}{ccc}	
		\includegraphics[width=.28\textwidth, height=.27\textwidth]{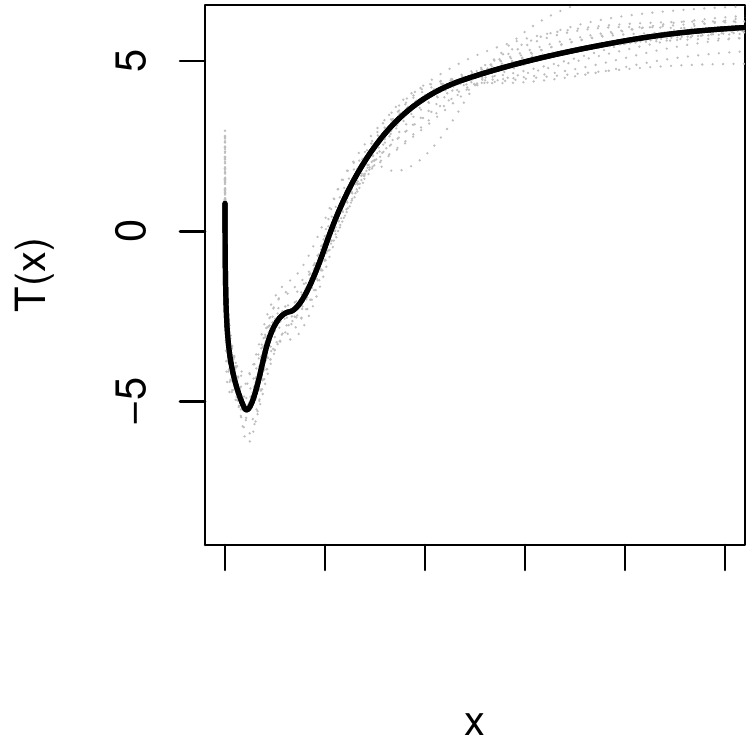}
		&	
		\includegraphics[width=.28\textwidth, height=.27\textwidth]{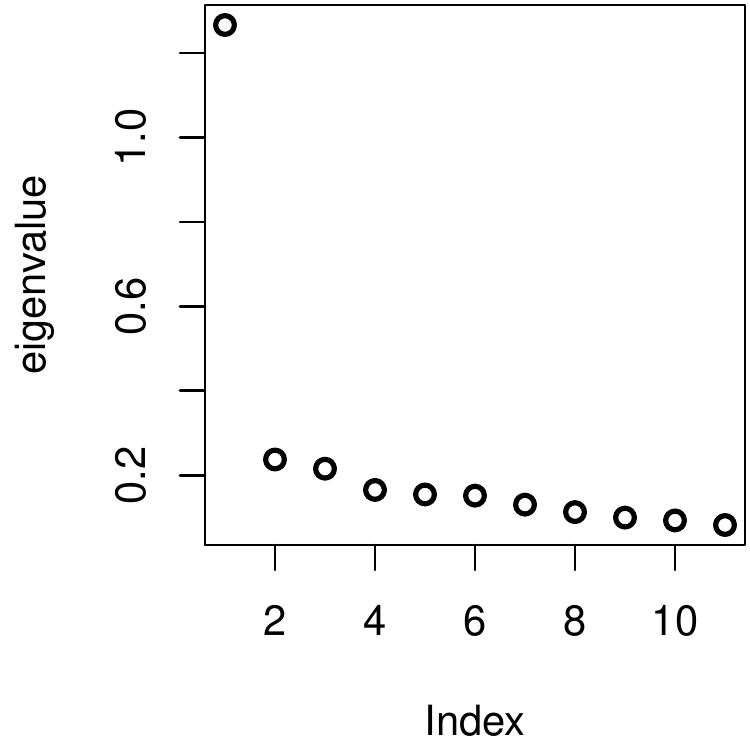}
		&	
		\includegraphics[width=.28\textwidth, height=.27\textwidth]{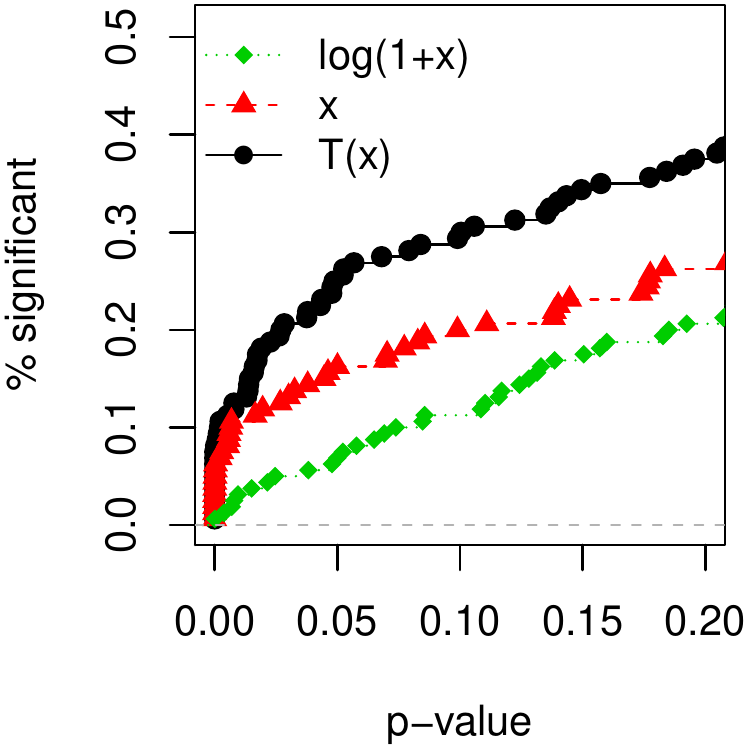}
	\end{tabular}
\vspace{-0.1cm}
	\caption{Left: Estimated sufficient statistics for time spent on the Facebook website. The black line shows the estimate using all the data. The gray lines show fits from bootstrapping. Center: Scree plot showing only one obvious sufficient statistic. Right: Cumulative distribution of $p$-values for tests using original values (red), log-transformed values (green), estimated sufficient statistic (black).}
	\label{fig:timeSpent}  			
\end{figure}

We show two metrics that can be of interest: number of photo likes and time spent on the Facebook website. These can be used as diagnostic measures for understanding user experience, for instance to detect bugs in the code or, for time spent, understanding if the page can be improved to provide users' relevant information more quickly. 
Both are measured at the user level over some period of time. 
For the photo likes metric, figure \ref{fig:photoLikes} shows that the estimated sufficient statistic has a shape similar to a log transformation.
This matches what one would expect to be a good transformation since the distribution of photo likes has a long tail. 
For the time spent metric, 
it also down-weights large values and effectively winsorizes them. 
Unlike the photo likes metric, the estimated sufficient statistic is clearly not monotonic, and we can only speculate as to the causes of the dip. However, the sufficient statistic leads to substantially more powerful tests as shown in figure \ref{fig:timeSpent}.
The proportion of significant tests at the 5\% level jumps from 15\% to 24\%. This difference has a p-value of 0.002 using McNemar's test. Interestingly, the time spent 
 sufficient statistic also demonstrates that 
the log transform is not always a good transform even if a distribution has fairly "heavy" tails. In this case the log transform performs worse than the untransformed data. It has thick tails as the standard deviation is high relative to the mean; the coefficient of variation is over 400\%.

\begin{figure}
	\centering
	\begin{tabular}{ccc}	
		\includegraphics[width=.28\textwidth, height=.28\textwidth]{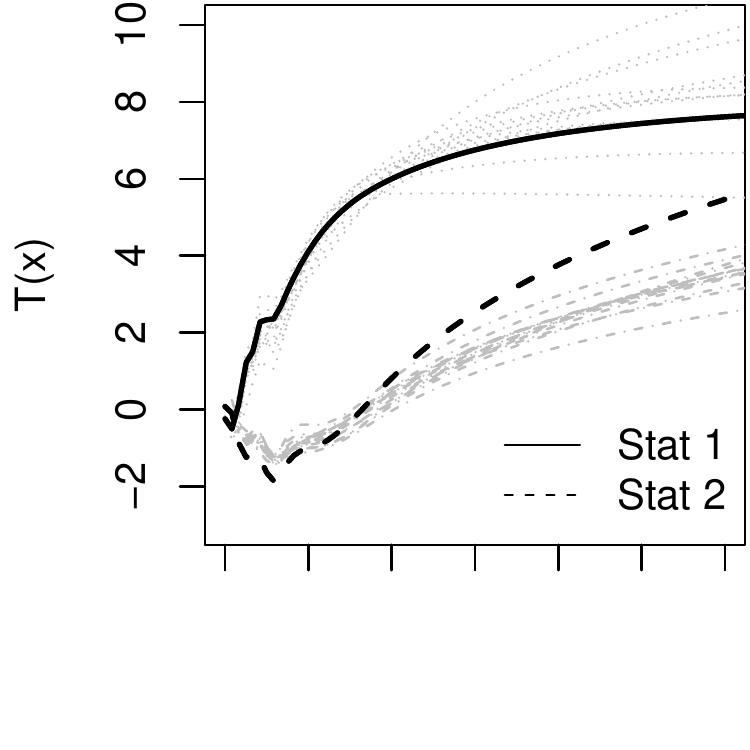}
		&	
		\includegraphics[width=.28\textwidth, height=.28\textwidth]{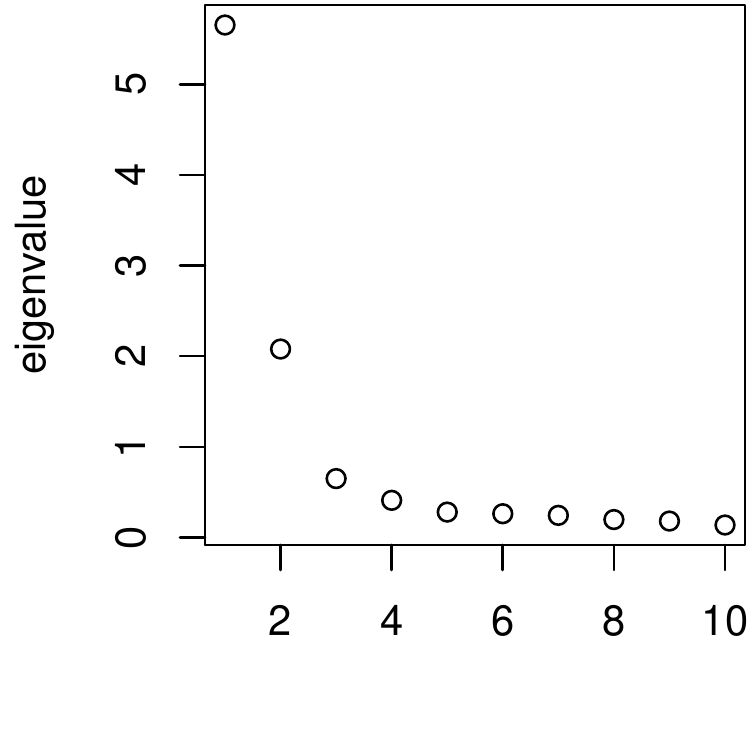}
\end{tabular}
	\vspace{-.5cm}
	\caption{Left: Estimated sufficient statistics for number of photo likes. The black lines shows the estimate using all the data. The grey lines show fits when bootstrapping the data. Center: Scree plot showing two possible sufficient statistics.}
	\label{fig:photoLikes}  
\end{figure}

\begin{table}
  \centering
  \begin{tabular}{cccc}
    Scenario & \multicolumn{3}{c}{Test Statistic}\\[5pt]
             & $X$ & $\log(1+X)$ & True $T(X)$ \\[3pt]
    
    Log-gamma & 0.08 & 0.85 & 1.01\\
    Laplace location & 0.53 & 0.71 & 1.07\\
    Facebook time spent & 0.32 & 0.04 & ---\\
    &&&
  \end{tabular}
  \caption{Asymptotic efficiency of tests relative to estimated sufficient statistic.}
  \label{tbl:RE}	
\end{table}

\section{Discussion and Conclusion}\label{sec:discussion}
We propose a method that non-parametrically learns a low-dimensional parametric family.
When applied to A/B testing, this allows combining many low powered experiments to find asymptotically efficient hypothesis tests and estimators. The efficacy of the method is demonstrated on real world A/B test data where it triples the effective sample size relative to a difference in means z-test.

One natural question for further study is how to choose the number of sufficient statistics to estimate. As shown in the scree plots in figures \ref{fig:timeSpent} and \ref{fig:photoLikes}, a good choice is often obvious. More systematically, one can split the data into test and training sets and evaluate a goodness of fit criterion such as a likelihood ratio.

One assumption that can be restrictive is that all treatment groups are drawn from a common user population. This is necessary for the spectral method that we have proposed in order to obtain an appropriate scaling matrix $W$. However, this assumption is not necessary if an exponential family model is learned directly. In this case, one fits the collection of models $\cM_i = p_{i,0}(x) e^{\theta_{i,j}'T(x) - \psi_i(x)}$ where $i$ denotes the subpopulation and $j$ the treatment group within that subpopulation. This can be fit by fitting Poisson generalized linear models that alternately estimate the parameters $\theta$ and the sufficient statistics $T$.

It would also be interesting to extend this work to handle multivariate responses. This may be done by extending the basis $S$ from a univariate to a multivariate basis. However, in order to well approximate any sufficient statistic, the number of needed basis elements may grow exponentially in the dimension; structural assumptions such as additivity or sparsity may be essential in this setting.

\textbf{Acknowledgement:}
We thank Facebook for allowing us to report results on their dataset. The authors are grateful for stimulating conversations with Jackson Gorham, Stefan Wager, Joe Romano, and Trevor Hastie. William Fithian was supported in part by the Gerald J. Lieberman Fellowship.

\bibliographystyle{plainnat}
\bibliography{biblio}

\appendix

\newcommand{\cN}{\mathcal{N}}
\newcommand{\tZ}{\widetilde{Z}}
\newcommand{\tX}{\widetilde{X}}
\newcommand{\tmu}{\tilde{\mu}}
\newcommand{\hrho}{\hat{\rho}}
\newcommand{\empCov}{\mathbb{C}ov}

\newcommand{\txi}{\tilde{\xi}}
\section{Preliminaries}

We begin by stating a lemma relating quadratic mean differentiability to differentiability of moments, which we will need in the proof.

\begin{lemma}\label{lem:qmd}
  Assume that $\cP = \{p_\theta:\; \theta\in \Omega \sub \R\}$ is QMD in a neighborhood of 0, and $T\in L^2(p_0)$ is the score function at $\theta=0$.
  Let $\mathcal{F}$ be a set of real-valued functions. If there exists $B > 0$ such that for all $f \in \mathcal{F}$, $\E_\theta f^2 < B$ for all $\theta$ in a neighborhood of 0, and $\E_0 |f^2 T|, \E_0 |fT^2|, \E_0 f^2T^2 < B$, then as $\theta\to 0$ we have 
  \begin{equation}\label{eq:qmdConclusion}
  \sup_{f \in \mathcal{F}} \left\vert \E_\theta - (\E_0 f + \theta \E_0 fT) \right\vert = o(|\theta|).
  \end{equation}
\end{lemma}

\begin{proof}

  Let $\xi_\theta = \sqrt{p_\theta}, \txi_\theta = \sqrt{p_0}(1+\theta T/2)$. QMD implies that $\| r_\theta \| = \|\xi_\theta - \txi_\theta\| = o(|\theta|)$ in norm, where the norm is defined with respect to the $L^2(\nu)$ inner product $\langle u, v\rangle = \int_{\cX}u(x)v(x)\,dx$. Then
\begin{align}
  \E_\theta f - \E_0 f
  &= \langle f (\txi_\theta + r_\theta), \txi_\theta + r_\theta\rangle 
    - \langle f\xi_0, \xi_0 \rangle\\
  &= \langle f \txi_\theta, \txi_\theta\rangle 
    - \langle f\xi_0, \xi_0 \rangle 
    + \langle r_\theta, f(r_\theta+2\txi_\theta)\rangle \\
  &= \langle\theta f T \xi_0, \xi_0\rangle 
    + \frac{1}{4}\langle f\theta^2T \xi_0, T\xi_0\rangle 
    + \langle r_\theta, f(\xi_\theta+\txi_\theta)\rangle \\
  &= \theta \E_0 fT + \frac{\theta^2}{4}\E_0 fT^2
    + \langle r_\theta, f(\xi_\theta+\txi_\theta)\rangle \\
  \left| \E_\theta f - (\E_0 f - \theta\E_0 fT) \right|
  &\leq \frac{\theta^2}{4}\E_0 fT^2
+ \|r_\theta \| \left|\E_\theta f^2 + \E_0 [f^2(1+\theta T/2)]\right|\\
  &\leq \frac{\theta^2 B}{4} + (2 +\theta/2) B \|r_\theta\|.
\end{align}
Because the final bound does not depend on $f$, the proof is complete.
\end{proof}

As an immediate corollary, suppose that $\sup_{x\in\cX} |S_i(x)| \leq B$ for all $i = 1,\ldots,d$. Then as $\theta \to 0$ we have for $c_i=\E_0 S_iT$ and $D_{ij}=\E_0 S_iS_jT$:
\[
\sup_{i} \left|\E_\theta S_i - (\E_0 S_i + \theta c_i)\right| \to 0 \qquad \text{ and } \qquad \sup_{i,j} \left|\E_\theta S_i S_j - (\E_0 S_i S_j + \theta D_{ij})\right| \to 0.
\]

\section{Proof of Theorem~\ref{thm:consistencyT}}

We restate Assumptions~\ref{assu:degrees}--\ref{assu:smallEffects} and Theorem~\ref{thm:consistencyT} here for convenience.

\begin{assumption}\label{assu:degrees}
$m, n_+-m \to \infty$ (the total degrees of freedom for $\hW$ and $\hA$ are growing)
\end{assumption}

\begin{assumption}\label{assu:signal}
$\frac{1}{m}\sum_{i=1}^m n_i(\theta_i-\bar\theta)^2 = \Omega(1)$ (the average signal strength does not shrink to 0)
\end{assumption}

\begin{assumption}\label{assu:smallEffects}
$\max_i |\theta_i| \to 0$ (effect sizes are small)
\end{assumption}

\setcounter{theorem}{0}
\begin{theorem}
  Under Assumptions~\ref{assu:degrees}--\ref{assu:smallEffects}, if we take $d \to\infty$ slowly enough, then $\hT(x)$ is consistent for $T(x)$ in $L^2(p_0)$ up to sign; that is, $\corr_0^2(\hT(X), T(X) | \hT) \toProb 1$.
\end{theorem}

In the statement of the Theorem, $X\sim p_0$ independently of the data points used to estimate $\hT$.

We will require a technical lemma that slightly generalizes Theorem 5.39 of~\citet{vershynin2010introduction}.
\begin{lemma}\label{lem:rmt}
Assume $A_1,\ldots,A_m$ are independent zero-mean $\sigma$-subgaussian random vectors in $\R^d$ with variances $W_1,\ldots, W_m$ respectively, and assume each $W_i \succeq \frac{1}{2} I_d$. Then with probability at least $1-2\exp(-C_1 t^2)$, we have
\[
\left\| \frac{1}{m} \sum_iA_iA_i' - W_i\right\|_\op \leq \max(\delta, \delta^2), \quad \text{ where } \delta = C_2 \sqrt{\frac{d}{m}} + \frac{t}{\sqrt{m}},
\]
where the constants $C_1$ and $C_2$ depend on $\sigma$.
\end{lemma}
The proof of Lemma~\ref{lem:rmt} is a straightforward modification of the argument in~\citet{vershynin2010introduction}. We include it for self-containment, after the proof of Theorem~\ref{thm:consistencyT}.

\begin{proof}[Proof (Theorem~\ref{thm:consistencyT}).]
  Because our method is invariant to invertible affine transformations of $S = (S_1,\ldots,S_d)$, we can assume without loss of generality that $W=I$ and $\mu_0 = 0$. If the user selects basis $\widetilde S$ with mean $\tilde \mu_0 \neq 0$ and invertible variance $\widetilde W \neq I_d$, we simply analyze the method using the affine transform $S = \widetilde W^{-1/2} (\widetilde S - \tilde \mu_0)$, which leads to the same estimator $\hT$. If $\widetilde W$ is not invertible we can first discard redundant basis functions and then transform to a lower-dimensional $S$. 

Note that if $\|\widetilde S(x)\|_\infty$ (which is chosen by the analyst) is bounded then $\|S(x)\|_{\infty}$ is bounded as well, by some constant $B_d$ that may grow with $d$; therefore $|\eta'S(x)|<\sqrt{d}B_d$ and $S(X_{i,j}) - \mu_i$ is sub-gaussian with variance proxy less than $\sigma_d = d B_d^2$.

  In that case 
  \begin{equation}
  A = \frac{\tau^2}{m-1} cc' + I_d, \quad f(\beta) = \frac{\beta'A\beta}{\beta'\beta} = \frac{\tau^2}{m-1} \left(\frac{\beta'c}{\|\beta\|}\right)^2 + 1, \quad \text{ and } \quad 
  f(\beta^*) = \frac{\tau^2\xi_d^2}{m-1} + 1,
  \end{equation}
  where $\|c\| = \xi_d^2$, the optimal correlation using the $d$-basis.

  Define the perturbations $Q=\hA-A$ and $P=\hW-I_d$. The result follows if we can prove that, for any fixed dimension $d$,
  \begin{align}
    \label{eq:Aconv} \|Q\|_\op = \|\hA - A\|_\op &= o_p\left(\frac{\tau^2}{m-1}\right) \\
    \label{eq:Wconv} \|P\|_\op = \|\hW - I_d\|_\op &= o_p(1) \quad\left[ = o_p\left(\frac{\tau^2}{m-1}\right)\right]  \end{align}
  
  To see why, note that, for $\|P\|_\op \leq 1/2$, we have
  \begin{align}
    \sup_\beta \hat f(\beta) - f(\beta)
    &= \sup_\beta \frac{\beta'A\beta + \beta'Q\beta}{\beta'\beta + \beta'P\beta} - \frac{\beta'A\beta}{\beta'\beta}\\
    &\leq \sup_\beta (1+2\|P\|_\op) \frac{\beta'A\beta + \beta'Q\beta}{\beta'\beta} - \frac{\beta'A\beta}{\beta'\beta}\\
    &\leq 2\|P\|_\op f(\beta^*) + (1+2\|P\|_\op)\|Q\|_\op\\
    &= o_p\left(\frac{\tau^2}{m-1}\right),
  \end{align}
  if \eqref{eq:Aconv} and \eqref{eq:Wconv} hold. There is a comparable bound for the other direction, leading to $\sup_\beta |\hat f(\beta) - f(\beta)| = o_p(\tau^2/(m-1))$. But because $f(\hat\beta) \geq f(\beta^*) - 2\sup|\hat f(\beta) - f(\beta)|$, we have
  \[
    \frac{\tau^2\rho_{\hat\beta}^2}{m-1} + 1 \geq \frac{\tau^2\rho_{\beta^*}^2}{m-1} + 1 + o_p\left(\frac{\tau^2}{m-1}\right)\quad \Rightarrow \quad \rho_{\hat\beta}^2 \geq \rho_{\beta^*}^2 + o_p(1).
  \]
  Because $\rho_{\beta^*}^2\to 1$ as $d\to\infty$, we have the desired consistency result as long as $d\to\infty$ slowly enough.

  Throughout the proof we will let
  \[
  \ep_{i,j} = S(X_{i,j})-\mu_i, \quad \bar\ep_i = \bar S_i - \mu_i = \frac{1}{n_i}\sum_i \ep_{i,j}, 
  \quad \text{ and }\quad \bar \ep = \frac{1}{n_+} \sum_i n_i \bar\ep_i = \frac{1}{n_+} \sum_{i,j} \ep_{i,j}.
  \]
  Recall the $\ep_{i,j}$ are all mutually independent with sub-gaussian norms $\|\ep_{i,j}\|_{\psi_2} \leq \sigma$; as a result $\|\bar\ep_i\|_{\psi_2} \leq \sigma/n_i$ and $\|\bar\ep\|_{\psi_2} = \sigma/n_+$.

  \paragraph{Step 1. Bound $\|\hA - A\|_\op$.}

  We begin by writing $\bar\mu = \frac{1}{n_+} \sum_i n_i\mu_i = \E \bar S$. Then recalling that $\bar S_i - \bar S = (\mu_i -\bar \mu) + (\bar\ep_i - \bar \ep)$, we can decompose
  \begin{align}
    (\bar S_i - \bar S)(\bar S_i - \bar S)'
    & = (\mu_i -\bar \mu)(\mu_i - \bar\mu)' + (\bar\ep_i - \bar \ep)(\bar\ep_i-\bar\ep)' \\
    &\quad+ \left[(\bar\ep_i - \bar \ep)(\mu_i -\bar \mu)' + (\mu_i -\bar \mu)(\bar\ep_i - \bar \ep)' \right].
  \end{align}

  We can decompose the error in $\hA$ as
  \begin{align}
    Q = \hA-A = 
    &\left[\frac{1}{m-1}\sum_i n_i (\bar S_i - \bar S)(\bar S_i - \bar S)'\right] - \left[\frac{\sum_in_i(\theta_i - \overline{\theta)}^2}{m-1}cc' + I_d\right]\\
    =& \frac{1}{m-1} \sum_i n_i \left[(\mu_i-\bar\mu)(\mu_i-\bar\mu)' - (\theta_i-\bar\theta)^2 cc'\right]\\
    & + \left[\frac{1}{m-1} \sum_i n_i \Var(\bar S_i - \bar S)\right] - I_d \\ 
    & + \frac{1}{m-1} \sum_i n_i \left[(\bar\ep_i - \bar \ep)(\bar\ep_i - \bar \ep)' - \Var(\bar S_i-\bar S)\right]\\
    & + \frac{1}{m-1} \sum_i n_i \left[(\bar\ep_i - \bar \ep)(\mu_i - \bar \mu)' + (\mu_i - \bar \mu)(\bar\ep_i - \bar \ep)'\right]
  \end{align}
  where $\bar \mu = \frac{1}{n_+} \sum_i n_i\mu_i$. Write $Q=Q_1+Q_2+Q_3+Q_4$, for the four sums above. Then $Q_1$ is a deterministic error from the approximation $\mu_i\approx \mu_0+\theta_i c = \theta_i c$, $Q_2$ is a deterministic error from the approximation $W_i\approx I_d$, and $Q_3$ and $Q_4$ are estimation errors.
  
  \paragraph{$Q_1$: Approximation Error for Means}
  
  Recalling Lemma~\ref{lem:qmd}, write $\mu_i = \theta_i c + \theta_i\zeta(\theta_i)$, where the remainder $\zeta(\theta_i)$ satisfies 
  \[
  \zeta^*(\delta) = \sup\{\|\zeta(\theta)\|:\; |\theta|\leq \delta\} \to 0 \quad \text{ as } \;\;\delta\to 0.
  \]
  Furthermore because $|\bar\theta| \leq \max_j|\theta_j|$, we have $\|\mu_i - \bar\mu - (\theta_i-\bar\theta) c\| \leq 2\zeta^*(\max_j|\theta_j|)$.

  As a result we obtain
  \begin{align}
    \|Q_1\|_{\op} &= \frac{1}{m-1}\sum_i n_i \left\| (\mu_i - \bar \mu)(\mu_i - \bar\mu)' - (\theta_i-\bar\theta)^2 cc'\right\|_{\op}\\
    \label{eq:squareBound}
                  &\leq \frac{1}{m-1}\sum_i n_i \left(2|\theta_i-\bar\theta|  + \|\mu_i - \bar\mu - (\theta_i-\bar\theta) c\|\right) \|\mu_i - \bar\mu - (\theta_i -\bar\theta)c\|\\
                  &\leq \frac{1}{m-1}\sum_i n_i \left(2|\theta_i-\bar\theta|  + |\theta_i-\bar\theta|2\zeta^*(\max_j|\theta_j|)\right) |\theta_i-\bar\theta| 2\zeta^*(\max_j|\theta_j|)\\
                  &= \frac{\tau^2}{m-1} 4\left[\zeta^*(\max_j|\theta_j|) + \zeta^*(\max_j|\theta_j|)^2\right]\\
    \label{eq:maxToTau}
                  &= o\left(\frac{\tau^2}{m-1}\right)
  \end{align}
  
  In \eqref{eq:squareBound} we use the fact that for $a,b\in\R^d$,
  \begin{align}
    \|aa'-bb'\|_{\op} 
    &= \frac{1}{2}\|(a+b)(a-b)' + (a-b)(a+b)'\|_{\op}\\
    &\leq (\|a\| + \|b\|) \|a - b\|\\
    &\leq \left(2\|b\| + \|a - b\|\right) \|a-b\|,
  \end{align}
  and that $\|c\| \leq 1$.
  
  \paragraph{$Q_2$: Approximation Error for Variances}
  To compute the variance of $\bar S_i - \bar S$, we note that
  \begin{equation}\label{eq:varDecomp}
    \sum_i n_i (\bar\ep_i-\bar \ep)(\bar\ep_i - \bar\ep)'
  = \sum_i n_i \bar\ep_i \bar\ep_i' - n_+ \bar \ep \bar \ep',
  \end{equation}
  leading to
  \begin{align}
    \sum_i n_i \Var(\bar S_i-\bar S)
    &= \sum_i n_i \Var(\bar S_i) - n_+ \Var(\bar S)\\
    &= \sum_i W_i - \frac{1}{n_+} \sum_i n_i W_i\\
    &= \sum_i\left(1-\frac{n_i}{n_+}\right) W_i.
  \end{align}
  Because we can also write $I_d = \frac{1}{m-1}\sum_i\left(1-\frac{n_i}{n_+}\right) I_d$, and noting that by Lemma~\ref{lem:qmd} we have $W_i = I_d + \theta_i D + o(|\theta_i|)$:
  \begin{align}
    Q_2 &= \left[\frac{1}{m-1}\sum_i\left(1-\frac{n_i}{n_+}\right)W_i\right] - I_d\\
        &= \frac{1}{m-1} \sum_i \left(1-\frac{n_i}{n_+}\right)(W_i-I_d) + I_d / m
  \end{align}

  Hence,
  \begin{align}
    \|Q_2\|_\op 
    &\leq \left|\frac{1}{m-1} \sum_i \left(1-\frac{n_i}{n_+}\right)(\theta_i D + o(|\theta_i|))\right|_\op + \|I_d / m\|_\op\\
    &= O(\max_i |\theta_i|) + 1/m\\
    &= o(1) = o\left( \frac{\tau^2}{m-1}\right),
  \end{align}
  
  \paragraph{$Q_3$: Estimation Error for Variances}
  Applying \eqref{eq:varDecomp} again, we can decompose $Q_3$ into two terms.
  \begin{align}
    Q_3 
    &= \frac{1}{m-1} \sum_i n_i \left[(\bar\ep_i- \bar\ep)(\bar\ep_i - \bar\ep)' - \Var(\bar S_i - \bar S)\right]\\
    &= \frac{1}{m-1} \sum_i (n_i \bar\ep_i\bar\ep_i' - W_i) \;\; - \;\; \frac{1}{m-1}\left[n_+\bar\ep \bar\ep' - \frac{1}{n_+}\sum_i n_iW_i\right].
  \end{align}
  Because $n_+\|\bar\ep\|^2 = \mathcal{O}_p(1)$, the operator norm of the second term is $\mathcal{O}_p(1/m)$. By Lemma~\ref{lem:rmt}, we have with probability at least $1-2\exp\{-C_1d\}$, \WFcomment{need to address dependence on $\sigma_d$}
  \[
  \left\|\frac{1}{m} \sum_i n_i \bar\ep_i\bar\ep_i' - W_i\right\|_\op \leq \max(\delta, \delta^2), 
\quad\text{ with }\quad \delta = (C_2+1)\sqrt{\frac{d}{m}}
  \]

  \paragraph{$Q_4$: Estimation Error, Cross-Term}
  
  We decompose $Q_4$ as $Q_4 = H + H'$, where
  \begin{align}
    H &= \frac{1}{m-1} \sum_i n_i (\bar\ep_i - \bar \ep) (\mu_i - \bar\mu)'\\
    &= \left(\frac{1}{m-1} \sum_i n_i \theta_i \bar \ep_i\right) (\mu_i - \bar \mu)'
       - \frac{n_+}{m-1} \bar \ep (\mu_i - \bar \mu)'.
  \end{align}
  
  Since $\sup_i \|\mu_i - \bar \mu\|_2 = O(\max_j |\theta_j|)$, $\left\| \sum_i n_i \bar \ep_i \right\|_2 = O_p(d)$, and $\| n_+ \bar \ep \|_2 = O_p(1)$  by Chebyschev's,
  \begin{align}
     \|H + H'\|_\op 
    &= O_p\left(d \max_j \theta_j^2 \right) + O_p\left( \max_j |\theta_j| \right) = o_p\left(\frac{\tau^2}{m-1}\right)
  \end{align}

  \paragraph{Step 2: Bound $\|\hW - I \|$.}
  
  We next bound $\|P\|_\op$. We first decompose $\hW$ into two variance components
  \begin{align}
    \hW &= \frac{1}{n_+-m}\sum_{i,j} (S(X_{i,j})-\bar S_i)(S(X_{i,j})-\bar S_i)'\\
    &= \frac{1}{n_+-m}\sum_{i,j} (\ep_{i,j}-\bar \ep_i)(\ep_{i,j}-\bar\ep_i)'\\
    &= \frac{1}{n_+-m}\sum_{i,j} \ep_{i,j}\ep_{i,j}' \;\;-\;\; \frac{1}{n_+-m} \sum_i n_i\bar\ep_i\bar\ep_i'.
  \end{align}
  This leads to a similar decomposition for the expectation:
  \[
  \E\hW = \bar W = \frac{1}{n_+ - m} \sum_{i,j} W_i \;\;-\;\; \frac{1}{n_+-m} \sum_i W_i.
  \]
  Matching the terms and applying Lemma~\ref{lem:rmt}, we obtain that with probability at least $1-4\exp\{-C_1d\}$,
  \[
  \|\hW - \bar W\|_\op \leq \frac{n_+}{n_+-m}\max(\delta_1,\delta_1^2) + \frac{m}{n_+-m}\max(\delta_2,\delta_2^2),
  \]
  with
  \[
  \delta_1 = (C_2+1)\sqrt{\frac{d}{n_+}},
  \quad \text{ and } \quad
  \delta_2 = (C_2+1)\sqrt{\frac{d}{m}}.
  \]
  Furthermore,
  \begin{align}
    \bar W - I 
    &= \frac{1}{n_+-m} \sum_i (n_i-1)(W_i - I)\\
    &= \frac{1}{n_+-m} \sum_i (n_i-1)\theta_i D - o(|\theta_i|)\\
    &= \mathcal{O}(\max |\theta_i|) \to 0
  \end{align}

\end{proof}

\section{Proof of Lemma~\ref{lem:rmt}}

We now turn to the proof of Lemma~\ref{lem:rmt}.

\begin{proof}[Proof (Lemma~\ref{lem:rmt}).]
   Write $Q = \frac{1}{m} \sum_i n_ia_ia_i' - W_i$ for the matrix whose norm we are bounding and $u = \max(\delta,\delta^2)$. Also, let $\|\cdot\|_{\psi_1}$ and $\|\cdot\|_{\psi_2}$ denote the sub-exponential and sub-gaussian norms respectively. 

   We repeat here the covering argument in~\citet{vershynin2010introduction}, but we must track the variance adjustment a bit more carefully in the concentration step. The proof proceeds in three parts. First, we discretize the unit sphere $S^{d-1}$ with a $\frac{1}{4}$-net $N$, which has cardinality $|N|\leq 9^d$. Second, we show concentration for each fixed vector $\beta\in S^{d-1}$. Third, we apply a union bound.
   
   \paragraph{Step 1. Approximation.} Using Lemma~5.4 of~\citet{vershynin2010introduction}, we can write
   \[
   \left\| Q \right\|_\op
   \; \leq \; 
   2 \max_{\beta\in N} \left| \beta'Q\beta\right|
   \; = \; 
   2 \max_{\beta\in N} \left| \frac{1}{m} \sum_i (a_i'\beta)^2 - \beta'W_i\beta \right|
   \]
   
   \paragraph{Step 2. Concentration.} Now fix any $\beta\in S^{d-1}$, and let $Z_i =  a_i'\beta$. The $Z_i$ are independent sub-gaussian random variables with $\E Z_i^2 = \beta'W_i\beta$ and $\|Z_i\|_{\psi_2} \leq \sigma$ ; hence, $V_i = Z_i^2 - \beta'W_i\beta$ are independent centered sub-exponential random variables with $\|V_i\|_{\psi_1} \leq 4\sigma^2$ and
   \[
   \sigma 
   \;\geq\;
   \|Z_i\|_{\psi_2}
   \;\geq\;
   \sqrt{\frac{\beta'W_i\beta}{2}}
   \;\geq\;
   \frac{1}{2}.
   \]

   Thus, applying a concentration inequality for sums of sub-exponential random variables, we obtain
   \[
   \P\left[ \left| \frac{1}{m} \sum_i V_i \right| \geq \frac{u}{2}\right] 
   \leq 2\exp\left[ - \frac{c_1}{\sigma^4} (C^2d + t^2)\right],
   \]
   for some constant $c_1$.
   
   \paragraph{Step 3. Union Bound.} Finally, we union bound over $\beta\in N$ to obtain
   \[
   \P\left[\max_{\beta\in N}\left\| \beta'Q\beta\right\| \geq \frac{u}{2}\right]
   \;\leq\; 9^d \cdot 2\exp\left[-\frac{c_1}{\sigma^4} (C^2d + t^2)\right] 
   \;\leq\; 2\exp\left[ -c_1t^2/\sigma^4\right],
   \]
   for $C=C_\sigma = \sigma^2 \sqrt{\log(9)/c_1}$.
\end{proof}

\section{Proof of Theorem~\ref{thm:optimality}}

\begin{proof}[Proof (Theorem~\ref{thm:optimality}).]
  To reduce notational clutter we re-name $P_{m+1}$ and $P_{m+2}$ as $P_1$ and $P_2$, with $P_i$ coming from density $p_{\theta_i}$, $\theta_i = t_i/\sqrt{n_i}$. Let $n$ be an index for the sequence of problems, with sample sizes $n_1(n),n_2(n)\to \infty$. Let $P_0$ denote the distribution represented by density $p_0$. Finally, assume without loss of generality that the Fisher information is 1 at $\theta=0$ (otherwise we could re-parameterize the family by scaling $\theta$).

  We will study the behavior of $\Delta_{\hT}$ for a fixed pair $t_1,t_2$ by relating it to the null case where $t_1=t_2=0$. Note that contiguity is satisfied in this problem: $P_1^{n_1}$ is contiguous to $P_0^{n_1}$ and $P_2^{n_2}$ is contiguous to $P_0^{n_2}$ because both are product measures under local alternatives at $\theta=0$; hence $P_1^{n_1} \times P_2^{n_2}$ is contiguous to $P_0^{n_1+n_2}$.

  Since the family is DQM at $\theta=0$ in the interior of $\Omega$, it follows that $\Delta_T \to \cN(0, 1)$ under the null. Furthermore, the test that rejects when $\Delta_T > \Phi^{-1}(1-\alpha)$, where $\Phi$ is the c.d.f. of a standard normal, is an asymptotically efficient test at level $\alpha$ by Theorem 15.14 and the obvious two-sample generalization to 15.5 in \cite{van2000asymptotic}. 
	
  The DQM condition implies local asymptotic normality for the sequence of experiments. Since $\corr_0(\hT_n(X), T(X)| \hT_n) \toProb 1$, it follows that $\Delta_{\hT_n} - \Delta_T \to 0$ in probability under the null as $n \to \infty$, and therefore also under the alternative by contiguity. As a result the test that rejects for $\Delta_{\hT_n} > \Phi^{-1}(1-\alpha)$ has the same asymptotic efficiency as $\Delta_T$.
\end{proof}

\end{document}